\documentclass[11pt]{article}
\usepackage[latin1]{inputenc}
\usepackage{graphicx}
\usepackage{epsfig}
\usepackage{color}
\usepackage{fullpage}
\usepackage{lscape}
\usepackage{verbatim}
\usepackage{amsthm, amssymb}
\usepackage{amsmath}
\newtheorem{Lem}{Lemma}
\newtheorem{theorem}{Theorem}
\newtheorem{Cor}{Corollary}

\def\polylog{\operatorname{polylog}}
\def\poly{\operatorname{poly}}

\title{Separator Theorems for Minor-Free and Shallow Minor-Free Graphs with Applications}
\author{Christian Wulff-Nilsen
        \footnote{School of Computer Science,
                  Carleton University,
                  \texttt{koolooz@diku.dk},
                  \texttt{http://cg.scs.carleton.ca/$_{\widetilde{~}}$cwn/}.
                  Research partially supported by NSERC and MRI.}}

\begin{document}

\maketitle
\begin{abstract}
Alon, Seymour, and Thomas generalized Lipton and Tarjan's planar separator theorem
and showed that a $K_h$-minor free graph with $n$ vertices has a separator of size
at most $h^{3/2}\sqrt n$. They gave an algorithm that, given a graph $G$ with $m$ edges and $n$ vertices and given an integer
$h\geq 1$, outputs in $O(\sqrt{hn}m)$ time such a separator or a $K_h$-minor of $G$.
Plotkin, Rao, and Smith gave an $O(hm\sqrt{n\log n})$ time algorithm to find a separator of size $O(h\sqrt{n\log n})$.
Kawarabayashi and Reed improved the bound on the size of the separator to $h\sqrt n$
and gave an algorithm that finds such a separator in $O(n^{1 + \epsilon})$ time for any constant $\epsilon > 0$, assuming
$h$ is constant. This algorithm has an
extremely large dependency on $h$ in the running time (some power tower of $h$ whose height is itself a function of $h$),
making it impractical even for small $h$. We are interested in a small polynomial time dependency on $h$ and
we show how to find an $O(h\sqrt{n\log n})$-size separator or report that $G$ has a $K_h$-minor in
$O(\poly(h)n^{5/4 + \epsilon})$ time for any constant
$\epsilon > 0$. We also present the first $O(\poly(h)n)$ time algorithm to find a separator of size $O(n^c)$ for a constant
$c < 1$. As corollaries of our results, we get improved algorithms for shortest paths and maximum matching.
Furthermore, for integers $\ell$ and $h$, we give an $O(m + n^{2 + \epsilon}/\ell)$ time algorithm that either produces
a $K_h$-minor of depth $O(\ell\log n)$ or a separator of size at most $O(n/\ell + \ell h^2\log n)$. This improves the
shallow minor algorithm of Plotkin, Rao, and Smith when $m = \Omega(n^{1 + \epsilon})$.
We get a similar running time improvement for an approximation algorithm for the problem of finding a largest $K_h$-minor in a
given graph.
\end{abstract}
\newpage

\section{Introduction}
Given a graph $G = (V,E)$ with a non-negative vertex weight function $w$, consider a partition of $V$ into subsets
$A$, $B$, and $C$ such that no edge joins a vertex of $A$ with a vertex of $B$ and such that $w(A),w(B)\leq cw(V)$ for
a constant $c < 1$. Then $C$ is a \emph{separator} of $G$ (w.r.t.\ $w$).

A graph $H$ is a \emph{minor} of a graph $G$ if $H$ can be obtained from a subgraph of $G$ by edge contraction. Note that for such a
contraction, vertices of $H$ correspond to disjoint connected subgraphs of $G$ and if $H$ is the complete graph $K_h$, there is at
least one edge in $G$ between each such pair of subgraphs. We refer to such a collection of subgraphs as an \emph{$H$-minor} of $G$.

A classical theorem by Lipton and Tarjan~\cite{SeparatorPlanar} states that every vertex-weighted planar graph of size $n$ has an
$O(\sqrt n)$-size separator. Alon, Seymour, and Thomas~\cite{SeparatorHMinor} generalized this by showing
the following: for a graph $G$ with $m$ edges and $n$ vertices and for an $h\in\mathbb N$,
there is an $O(\sqrt{hn}m)$ time algorithm that either produces a $K_h$-minor of $G$ or a separator of size at most $h^{3/2}\sqrt n$.

Kawarabayashi and Reed~\cite{SepOpt} improved the size bound to
$h\sqrt n$ and gave an $O(n^{1 + \epsilon})$ time algorithm
for any constant $\epsilon > 0$, assuming $h$ is fixed. The hidden dependency on
$h$ in the running time is huge, in fact some power tower of $h$ whose height is itself a function of
this parameter\footnote{Through internal communication
with K.~Kawarabayashi; the dependency is not stated in~\cite{SepOpt}, where $h$ is assumed to be a constant in the analysis.}.
Reed and Wood~\cite{FastSeparatorHMinor} gave
a trade-off between running time and the separator size in a minor-free graph:
for any $\gamma\in[0,\frac 1 2]$, there is an $O(2^{(3h^2 + 7h - 3)/2}n^{1 + \gamma} + hm)$ time algorithm giving either
a $K_h$-minor of $G$ or a separator of size at most $2^{(h^2 + 3h + 1)/2}n^{(2 - \gamma)/2}$.
In particular, this gives an $O(n)$ time algorithm to find a separator of size $O(n^{2/3})$ for fixed $h$.
In general, neither of these two algorithms run in polynomial time in the size of the input.

\subsection{Separators for shallow minor-free graphs}
Plotkin, Rao, and Smith~\cite{ShallowMinor} considered a larger class of graphs, those excluding a shallow/limited depth minor.
Given graphs $G$ and $H$ and
an integer $L$, $H$ is a \emph{depth $L$-minor} or a \emph{minor of depth $L$} of $G$ if there exists an $H$-minor of $G$ in which
each subgraph has diameter at most $L$.
A motivation for these graphs can be found in geometry. As shown
in~\cite{ShallowMinor}, $d$-dimensional simplicial graphs of bounded aspect ratio exclude $K_h$ as a depth $L$-minor if
$h = \Omega(L^{d-1})$ for constant $d$. This illustrates that the dependency on $h$ in the time and in the separator size
should be taken into consideration. Graphs excluding minors of non-constant size were also considered in~\cite{Butterfly}.

Plotkin, Rao, and Smith showed that
given a graph with $m$ edges and $n$ vertices and given integers $\ell$ and $h$, there is an $O(mn/\ell)$ time algorithm
that either produces a $K_h$-minor of depth $O(\ell\log n)$ or a separator of size $O(n/\ell + \ell h^2\log n)$.
A suitable choice of $\ell$ gives an $O(hm\sqrt{n\log n})$ time algorithm that either produces a $K_h$-minor or
a separator of size $O(h\sqrt{n\log n})$.

Let us sketch their algorithm as we will consider it in this paper.
A three-way partition $(M,V_r,V')$ of $V(G)$ is kept. Subset $V'$ contains unprocessed vertices
(initially, $V' = V$), $V_r$ contains processed vertices, and $M$ is spanned by trees forming a $K_p$-minor of depth $O(\ell\log n)$
of $G$, for some $p\leq h$. The algorithm is iterative and in each iteration, the subgraph induced by $V'$ is considered.
Either a tree of depth $O(\ell\log n)$ is formed that can be
added to $M$ (thereby increasing $p$ by one) or a set $S\subseteq V'$ is formed which can be added to $V_r$ without introducing
too much vertex weight to $V_r$. The algorithm terminates when the weight of $V'$ is also small. Furthermore, at each step, the size
of the
set $B$ of vertices of $V'$ incident to $V_r$ is small so when all vertices have been processed, $M\cup B$ will form a small separator.
However, it may happen that $p = h$ during the course of the algorithm in which case a $K_h$-minor of depth $O(\ell\log n)$ has been
identified.

\subsection{Our separator theorems}
We will show how to find small separators in minor-free and shallow-minor free graphs $G$ in time which is small
in the size of $G$ and in $\ell$ and $h$.
Our overall approach is the same as that in~\cite{ShallowMinor} but we introduce several new ideas to get a more
efficient algorithm.
The first idea is quite simple. We maintain the subgraph of $G$ induced by $V'$ with a dynamic spanner of constant stretch. This
gives a speed-up for all graphs with
$m = O(n^{1 + \epsilon})$ for arbitrarily small constant $\epsilon > 0$. The same idea yields
a faster approximation algorithm than that of Alon, Lingas, and Wahlen~\cite{ApproxMinor} for the problem of computing a
largest $K_h$-minor in a given graph.

To get a further improvement for minor-free graphs, we use a technique similar to bootstrapping in, say, compiler design:
in a first step,
we construct large separators fast and use them to build a certain clustering of $G$. We preprocess a data structure for this
clustering and then use it in a second step
to build a small separator fast (or report the existence of a $K_h$-minor). For $\ell = \Omega(\sqrt n/h\sqrt{\log n})$,
this gives an $O(\ell h^2\log n)$-size separator in $O(\poly(h)n^{3/2 + 5\epsilon}/\ell^{1/2 - \epsilon})$ time
for any constant $\epsilon > 0$. In particular, we get a separator of size $O(h\sqrt{n\log n})$ in
$O(\poly(h)n^{5/4 + \epsilon})$ time.

An important corollary is a separator theorem with $O(\poly(h)n)$ running time giving a separator of size
$O(n^{4/5 + \epsilon})$
for any constant $\epsilon > 0$ (no hidden dependency on $h$ in the size). Previously, no $O(\poly(h)n)$ time
algorithm was known that gives a separator of size $O(n^\delta)$ for any constant $\delta < 1$.

\subsection{Applications}
Henzinger, Klein, Rao, and Subramanian~\cite{SSSPPlanar} showed that single source shortest paths in planar graphs with non-negative
edge weights can be found in linear time. They stated that their algorithm generalizes to classes of graphs for which a separator of
size $O(n^\delta)$, $\delta < 1$ a constant, can be found in linear time. Since our $O(\poly(h)n)$ time algorithm gives a separator
of size $O(n^{4/5 + \epsilon})$, we should therefore have an $O(\poly(h)n)$ time algorithm for shortest paths in minor-free graphs
with non-negative
edge weights. Unfortunately there is a problem. As noted by Tazari and M\"{u}ller-Hannemann~\cite{LinTimeSSSPHMinor}, the algorithm
in~\cite{SSSPPlanar} requires the graph to have constant degree. This can be assumed w.l.o.g.\ for planar graphs using vertex
splitting but this is not true for minor-free
graphs in general. Tazari and M\"{u}ller-Hannemann found a way around this problem and got an algorithm for minor-free graphs
with $O(n)$ running time for fixed $h$.
However, their algorithm has exponential dependency on $h^2$ since they rely on the separator theorem of Reed and
Wood~\cite{FastSeparatorHMinor}. Our algorithm needs the bounded degree assumption but
runs in optimal $O(n)$ time (no dependency on $h$) if $h = O(n^c)$ for a constant $c > 0$; for larger $h$-values than this, the
problem is not so interesting since we can apply Dijkstra which runs in $O(n\log n)$ time (since the graph is sparse).

Furthermore, we apply our separator theorem to get an
$\tilde O(\poly(h)n^{4/3}\log L)$\footnote{Throughout the paper, we use $\tilde O$- and $\tilde\Omega$-notation when
suppressing $\log n$-factors, so e.g.\ $\tilde O(n)$ means $O(n\polylog n)$.} time algorithm for single-source
shortest paths in minor-free graphs with negative edge weights, where $L$ is the absolute value of the smallest edge weight. This
improves Yuster's algorithm~\cite{Yuster} which has roughly $O(\poly(h)n^{1.392}\log L)$ running time and matches an earlier bound
for planar graphs in~\cite{SSSPPlanar} up to logarithmic factors\footnote{The time bound for planar graphs has
since been improved in a series of papers~\cite{Fakcharoenphol,SSSPPlanarNeg,SSSPPlanarNeg2}; the current best bound is
$O(n\log^2n/\log\log n)$.}. The recent separator theorem of Kawarabayashi and
Reed~\cite{SepOpt} gives the same time bound in terms of $n$ and $L$ but its dependency on $h$ is huge as mentioned above.

Finally, we obtain a faster algorithm for maximum matching in minor-free graphs. We get a time bound of roughly
$O(\poly(h)n^{1.239})$ which improves the $O(\poly(h)n^{1.326})$ time bound of Yuster and Zwick~\cite{MaxMatching}. A slightly better
time bound of $O(n^{\omega/2}) < O(n^{1.188})$ can be obtained using the separator theorem of Kawarabayashi and Reed~\cite{SepOpt}
but again with a very large dependency on $h$.


\section{Definitions and Notation}\label{sec:Prelim}
For a graph $G$, if we do not name its vertex and edge sets, we shall refer to them as $V(G)$ and $E(G)$, respectively.
We define $|G| = |V(G)| + |E(G)|$.
For a subset $X$ of $V(G)$, we denote by $N_G(X)$ the set of vertices of distance at most $1$ from $X$ in $G$
(note that $X\subseteq N_G(X)$). For an integer $\delta\geq 1$, we define $N_G^\delta(X) = N_G(X)$ if
$\delta = 1$ and $N_G^\delta(X) = N_G(N_G^{\delta-1}(X))$ otherwise. When there is no confusion, we write $N(X)$ and $N^\delta(X)$
instead of $N_G(X)$ and $N_G^\delta(X)$, respectively. We let $G[X]$ denote the subgraph of $G$ induced by $X$.
If $G$ is edge-weighted, we define $d_G(s,t)$ to be the distance from $u$ to $v$ in $G$ w.r.t.\ these edge weights.

For a real value $\delta\geq 1$, a \emph{$\delta$-spanner} of an edge-weighted graph $G$ is a subgraph $H$ of $G$ spanning $V(G)$
such that for any $u,v\in V(G)$, $d_H(u,v)\leq\delta d_G(u,v)$. We call $\delta$ the \emph{stretch} of $H$ (in $G$).

\section{Shallow Minors}\label{sec:ShallowMinors}
In this section, we give an improved separator theorem for shallow minor-free graphs as well as two applications of this result.
\begin{theorem}\label{Thm:ShallowMinor}
Given a vertex-weighted graph $G = (V,E)$ with $m$ edges and $n$ vertices, given $\ell,h\in\mathbb N$, and given a constant
$\epsilon > 0$, there is an $O(m + n^{2 + \epsilon}/\ell)$ time algorithm that either produces a $K_h$-minor of $G$ of
depth $O(\ell\log n)$ or finds a separator of size $O(n/\ell + \ell h^2\log n)$.
\end{theorem}
The proof is similar to that in~\cite{ShallowMinor}; see details in Appendix~\ref{sec:ProofShallowMinorAll}.
The idea is to maintain a sparse approximate representation of the subgraph induced by the unprocessed set of vertices $V'$ using
a dynamic $O(1)$-spanner. This will only increase the size of the separator by a constant factor.
However, using a spanner creates a problem during the course of the algorithm as edges
of the underlying graph may cross over the partially built separator $B$ in the full graph even if they do not do so in the spanner.
We can avoid this by picking $B$ ``thick'' enough.
\begin{Cor}\label{Cor:ShallowMinor}
Given a vertex-weighted graph $G$ with $m$ edges and $n$ vertices, given $h\in\mathbb N$, and given a constant $\epsilon > 0$, there
is an $O(m + hn^{3/2 + \epsilon})$ time algorithm that either produces a $K_h$-minor of $G$
or finds a separator of size $O(h\sqrt{n\log n})$.
\end{Cor}
We get the following improvement of a result in~\cite{ApproxMinor} when $m = \Omega(n^{1 + \epsilon})$.
\begin{theorem}\label{Thm:ApproxMinor}
There is an $O(\sqrt n\log^{3/2}n)$ approximation algorithm for the problem of finding a largest $K_h$-minor in a graph with
$m$ edges and $n$ vertices running in time $O(m\log n + hn^{3/2 + \epsilon})$ for any constant $\epsilon > 0$.
\end{theorem}
\begin{proof}
As shown in~\cite{ApproxMinor},
the problem reduces to applying a separator theorem on the input graph and then recursing on the connected components. Plugging
Corollary~\ref{Cor:ShallowMinor} into that lemma instead of the separator theorem in~\cite{ShallowMinor} then gives the result.
\end{proof}

\section{Nested $r$-Clustering}\label{sec:NestedDiv}

For a connected graph $G$, a \emph{cluster} (of $G$) is a connected subgraph of $G$. A \emph{clustering} (of $G$) is a
partition $\mathcal C$ of $G$ into pairwise edge-disjoint clusters. A \emph{boundary vertex} of a cluster $C\in\mathcal C$ is
a vertex that $C$ shares with other clusters in $\mathcal C$. All other vertices of $C$ are
\emph{interior vertices} of $C$. For a subgraph $C'$ of $C$, we let $\delta C'$ denote the set of boundary vertices of $C$
contained in $C'$ and we refer to them as the boundary vertices of $C'$ (w.r.t.\ $C$).

Let $n$ be the number of vertices of $G$. For a parameter $r > 0$, an \emph{$r$-clustering} (of $G$) is
a clustering with clusters having a total of $\tilde O(hn/\sqrt r)$ boundary vertices (counted with multiplicity) and each
containing at most $r$ vertices and $\tilde O(h\sqrt r)$ boundary vertices. Observe that the total vertex size of clusters
in an $r$-clustering is $\tilde O(hn/\sqrt r + n)$. Since $G$ is connected, each cluster contains at least one boundary vertex unless
there is just one cluster. Hence, the number of clusters is $\tilde O(hn/\sqrt r)$.

Define a \emph{nested $r$-clustering} of $G$ as follows. Start with an $r$-clustering. Partition each cluster $C$
with a separator of size $\tilde O(h\sqrt{|V(C)|})$ and recurse until clusters
consisting of single edges are obtained. We do it in such a way that both the sizes of clusters and their number of boundary vertices
go down geometrically through the recursion. This is possible with any standard separator theorem due to the
result by Djidjev and Gilbert~\cite{MultiWeightSep} (see Theorem $5$ in their paper).

The nested $r$-clustering is the collection $\mathcal C$ of all clusters obtained by this recursive procedure. There are $O(\log r)$
recursion levels and we refer to clusters of $\mathcal C$ at level $i\geq 1$ as \emph{level $i$-clusters} of
$\mathcal C$. Note that level $1$-clusters are those belonging to the original $r$-clustering. We define parent-child
relationships between clusters of $\mathcal C$ in the natural way defined by the recursion.
\begin{Lem}\label{Lem:rDiv}
Let $G$ be a vertex-weighted graph with $m$ edges and $n$ vertices, let $h\in\mathbb N$, and let $\epsilon > 0$
be a constant. For any parameter $r\in(Ch^2\log n,n)$ for a sufficiently large constant $C$, there is an algorithm with
$O(m\log n + hn^{1 + \epsilon}\sqrt r)$ running time that finds a $K_h$-minor or a nested $r$-clustering of $G$.
\end{Lem}
The proof can be found in Appendix~\ref{sec:ProofrDiv}. The main idea is to apply Theorem~\ref{Thm:ShallowMinor} with a large
value of $\ell$ to find a big separator fast and then recurse until an $r$-clustering is found. The slow
Corollary~\ref{Cor:ShallowMinor} can then be applied recursively to these small clusters instead of the entire graph.

\section{Minor-free Graphs}\label{sec:SepThm}
In this section, we give our separator theorem for minor-free graphs. The speed-up comes from a certain canonical decomposition of our
graph, which we present in Section~\ref{sec:DecomposeSubgraphs}, as well as from
Lemma~\ref{Lem:TS} below which we prove in Sections~\ref{sec:ShortestPaths} and~\ref{sec:SepSet}.
The following lemma from~\cite{SparsityHMinor1,SparsityHMinor2}, allows us to assume that our graph
$G$ is sparse (modulo a small multiplicative dependency on $h$).
\begin{Lem}\label{Lem:Sparse}
For any graph $H$, if $H$ does not contain $K_h$ as a minor then $|E(H)| = O(h\sqrt{\log h}|V(H)|)$.
\end{Lem}
We shall assume that any subgraph $H$ of $G$ considered by the algorithm satisfies the bound on the number of edges in
Lemma~\ref{Lem:Sparse} since otherwise $H$, and hence $G$, contains $K_h$ as a minor.

We do not depart from the overall iterative strategy of Plotkin, Rao, and Smith.
We maintain a four-way partition $(V_r, M, B, V')$ during the iterative algorithm which is
initialized to $V' = V$ and $M = B = V_r = \emptyset$ and which satisfies the three invariants in the proof of
Theorem~\ref{Thm:ShallowMinor}. Set $M$ can be partitioned into $p\leq h$ subsets
$\mathcal A_1,\ldots,\mathcal A_p$ with the same properties as in that proof.

In all the following, let $\ell = \Omega(\sqrt n/(h\sqrt{\log n}))$ and let $r = n/(h^2\ell)$. Before the first iteration, we
compute a nested $r$-clustering $\mathcal C$ of $G$ using Lemma~\ref{Lem:rDiv}. For each cluster
$C\in\mathcal C$, we check if $w(C)\geq \frac 2 3 w(V)$. If so, we return $V(C)$ since it is a separator and its size is
bounded by the size of a cluster in $\mathcal C$ which is $O(r) = O(n/(h^2\ell)) = O(\ell\log n)$ since
$\ell = \Omega(\sqrt n/(h\sqrt{\log n}))$.

Otherwise, we start the iterative algorithm.
In each iteration, we do as follows. If $p = h$ return the collection $\mathcal A_1,\ldots,\mathcal A_h$ as a
$K_h$-minor of $G$. If $w(V(C))\leq \frac 2 3 w(V)$ for each connected component $C$ of $G[V']$ output the separator $M\cup B$.

Otherwise, let $G'$ denote the unique connected component of $G[V']$ of weight larger than $\frac 2 3 w(V)$ and
let $A_i = N_G(\mathcal A_i)\cap V(G')$ for $i = 1,\ldots,p$. By the above, $G'$ is not fully contained in a
cluster so it must contain at least one boundary vertex from $\mathcal C$. This observation will be needed later when we prove
the following result which is similar to Lemma 2.5 in~\cite{ShallowMinor}.
\begin{Lem}\label{Lem:TS}
In each iteration of the algorithm above, either there is
\begin{enumerate}
\item an index $i$ such that $A_i = \emptyset$,
\item a tree $T$ in $G'$ of depth $O(\ell\log n)$ and size $O(h\ell\log n)$ with
      $V(T)\cap A_i\neq\emptyset$ for $i = 1,\ldots,p$, or
\item a set $S\subseteq V(G')$ such that
\begin{enumerate}
  \item $|N(S)\cap V(G')\setminus S|\leq\min\{|S|,|V(G')\setminus S|\}/\ell$, and
  \item $|N(V(G')\setminus S)\cap S|\leq\min\{|S|,|V(G')\setminus S|\}/\ell$.
\end{enumerate}
\end{enumerate}
For any constant $\epsilon > 0$, there is a dynamic algorithm that
finds such an index $i$ or such a tree $T$ in $\tilde O((h^2\sqrt{\ell n})^{1 + \epsilon})$ time or such a set $S$
in $O((h^2\sqrt{\ell n})^{1 + \epsilon} + \min\{|G'[N_{G'}(S)]|,|G'[N_{G'}(V(G')\setminus S)]|\})$ time. Additional time required
over all iterations is $\tilde O(h^2n^{3/2}/\sqrt{\ell} + n^{3/2 + \epsilon}/\sqrt\ell  + h^4n)$.
\end{Lem}
If the lemma returns an index $i$, we move $\mathcal A_i$ from $M$ to $V_r$. If a
tree $T$ is found, we extend it to one of size $\Theta(\min\{h\ell\log n,|V(G')|\})$ if needed, move its vertex set
$\mathcal A_{p+1}$ to $M$, and update $p := p + 1$. If we find a set $S$, we identify the set $S'$ that has
least weight among $S$ and $V(G')\setminus S$ and move it from $V'$ to $V_r$ and move
$N_{G'}(S')\setminus S'$ from $V'$ to $B$. We then start the next iteration.

Correctness of this algorithm follows from the analysis of Plotkin, Rao, and Smith and our analysis in the proof of
Theorem~\ref{Thm:ShallowMinor}.
Let us bound running time. There are $O(n/\ell)$ iterations in total so the total time to
find indices $i$ and trees $T$ is
$\tilde O(h^{2 + 2\epsilon}n^{3/2 + \epsilon/2}/\ell^{1/2 - \epsilon/2}) =
\tilde O(h^2n^{3/2 + 5\epsilon/2}/\ell^{1/2 - \epsilon/2})$. If we find a set
$S$, we need to compute $w(S)$ and
$w(V(G')\setminus S)$ in order to identify $S'$. In Section~\ref{sec:DecomposeSubgraphs}, we show how to report the weight of $G'$ in
$\tilde O(h^2\sqrt{\ell n})$ time. We can identify the weight of one of the sets $S$ and $V(G')\setminus S$ in
$O(\min\{|G'[S]|,|G'[V(G')\setminus S]|\})$ time and from this and from $w(V(G'))$ we can identify the weight of
the other set and hence find $S'$. The time for this is dominated by the $O(|G'[N_{G'}(S')]|)$ time to find $N_{G'}(S')$ which we can
charge to $G'[N_{G'}(S')]$ being deleted from $G[V']$ since we never add anything to $G[V']$. We conclude that it takes a total of
$O(m + h^2n^{3/2}/\sqrt\ell)$ time to find sets $S'$ and to move them and $N_{G'}(S')\setminus S'$ to $V_r$ and $B$, respectively.
This is $\tilde O(h^2n^{3/2}/\sqrt\ell)$ by Lemma~\ref{Lem:Sparse}.
The total time to move $\mathcal A_i$-sets from $M$ to $V_r$ is within this bound as well.

The additional processing time required in Lemma~\ref{Lem:TS} includes updates to an underlying data
structure when sets $V_r$, $M$, $B$, and $V'$ change. The nested $r$-clustering $\mathcal C$ constitutes part of this data structure
which we later describe in detail.

Our separator theorem for minor-free graphs now follows, given the assumptions above.
\begin{theorem}\label{theorem:Sep}
Given a vertex-weighted graph $G$ with $n$ vertices, given an integer $h$, and given a constant $\epsilon > 0$, there
is an $O(h^2n^{3/2 + \epsilon}/\ell^{1/2} + h^4n\polylog n)$ time algorithm that either
reports that $G$ has a $K_h$-minor or finds a separator of size $O(\ell h^2\log n)$.
\end{theorem}

\begin{Cor}\label{Cor:Sep}
Given a vertex-weighted graph $G$ with $n$ vertices, given an integer $h$, and given a constant $\epsilon > 0$, there
is an $O(h^{5/2}n^{5/4 + \epsilon} + h^4n\polylog n)$ time algorithm that either
reports that $G$ has a $K_h$-minor or finds a separator of size $O(h\sqrt{n\log n})$.
\end{Cor}

\section{Decomposing Subgraphs}\label{sec:DecomposeSubgraphs}
The algorithm of Lemma~\ref{Lem:TS} needs a compact representation of certain connected components of $G\setminus X$.
In this section, we give such a representation using the nested $r$-clustering of Section~\ref{sec:NestedDiv}.

As in Section~\ref{sec:SepThm}, let $\mathcal C$ denote the nested $r$-clustering of $G$. We need the following two
lemmas, the proofs of which can be found in Appendix~\ref{sec:ProofR} and~\ref{sec:ProofRX}.
\begin{Lem}\label{Lem:R}
$\sum_{C\in\mathcal C}|C||\delta C| = \tilde O(hn^{3/2}/\sqrt{\ell})$ and
$\sum_{C\in\mathcal C}|\delta C|^3 = \tilde O(h^2n^{3/2}/\sqrt{\ell} + h^4n)$.
\end{Lem}

For a given set $X\subseteq V(G)$, define a subset
$\mathcal C_X\subseteq\mathcal C$, obtained by the following procedure. Initialize $\mathcal C_X$ to be the set of
level $1$-clusters of $\mathcal C$. As long as there exists a cluster of $\mathcal C_X$ with
at least one interior vertex belonging to $X$, replace it in $\mathcal C_X$ by its child clusters from $\mathcal C$.
\begin{Lem}\label{Lem:RX}
If $|X| = \tilde O(h^2\ell)$ then $\mathcal C_X$ consists of clusters sharing only boundary vertices, all vertices of $X$ are
boundary vertices in $\mathcal C_X$, and $\sum_{C\in\mathcal C_X}|\delta C| = \tilde O(h^2\sqrt{\ell n})$.
\end{Lem}

In the algorithm of Section~\ref{sec:SepThm}, observe that $G[V']$ is the union of connected components of $G\setminus X$, where
$X = M\cup B$, and that $X$ changes during the course of the algorithm. Let us therefore consider the following
dynamic scenario. Suppose that vertices of $G$ can be in two states, \emph{active} and \emph{passive}. Initially,
all vertices are passive. A vertex can change from passive to active and from active to passive at most once.
Furthermore, at any given point in time, only $O(h^2\ell\log n)$ vertices are active. If we let $X$ be the set of active vertices,
our algorithm from Section~\ref{sec:SepThm} satisfies these properties since $|X| = O(n/\ell + h^2\ell\log n)$ and
$n/\ell = O(h^2\ell\log n)$.
In the following, we will show how to maintain, in this dynamic scenario, a compact representation of those connected components of
$G\setminus X$ that contain at least one boundary vertex from $\mathcal C_X$, as well as their vertex weights.


\paragraph{$X$-clusters:}
We will maintain some information for each cluster $C\in\mathcal C$.
As $X$ changes, the state of boundary vertices of $C$ may change between active and passive. We will refer
to those connected components of $C\setminus(\delta C\cap X)$ that contain at least one (passive) vertex of $\delta C$
as the \emph{$X$-clusters} of $C$. Since by assumption, any vertex
can change its passive/active state at most twice, the total number of active/passive updates in $\delta C$ is
$O(|\delta C|)$. After each such update, we compute the weights of the new $X$-clusters of $C$. This can be done in
$O(|C|)$ time for each update for a total of $O(|C||\delta C|)$.
By Lemma~\ref{Lem:R}, this is $\tilde O(hn^{3/2}/\sqrt{\ell})$ over all clusters $C\in\mathcal C$.

In the following, when we consider a set
$\mathcal K$ of clusters, we will say that a boundary vertex resp.\ an $X$-cluster is in $\mathcal K$ if it is a boundary vertex
resp.\ an $X$-cluster of a cluster in $\mathcal K$.

\paragraph{Decomposition:}
Now consider the set $\mathcal H$ of connected components of $G\setminus X$ that contain boundary vertices from $\mathcal C_X$.
We can obtain $\mathcal C_X$ from $\mathcal C$ in time proportional to
$|\mathcal C_X|$ which is $\tilde O(h^2\sqrt{\ell n})$ by Lemma~\ref{Lem:RX}. Each component of $\mathcal H$ can be expressed as
the union of $X$-clusters in $\mathcal C_X$. By precomputing adjacency information between $X$-clusters sharing
boundary vertices, we can identify the $X$-clusters forming each component of $\mathcal H$ in time proportional to
$O(|\mathcal C_X|)$ with, say, a DFS algorithm. Since we maintain the vertex weight of each $X$-cluster and since by
Lemma~\ref{Lem:RX} the total number of boundary vertices (and hence $X$-clusters) in $\mathcal C_X$ is $\tilde O(h^2\sqrt{\ell n})$,
we can obtain the vertex weight
of each component of $\mathcal H$ within the same time bound. Note that when adding up the
weights of the $X$-clusters forming a component of $\mathcal H$, we overcount the contribution from vertices belonging to more
than one $X$-cluster. Since each such vertex is a boundary vertex, we can easily take care of this in time proportional to the number
of boundary vertices in $\mathcal C_X$ which is $\tilde O(h^2\sqrt{\ell n})$.
Combining these results with Lemmas~\ref{Lem:rDiv} and~\ref{Lem:Sparse}, we get the following.
\begin{Lem}\label{Lem:DecompSubgraphs}
Let $\epsilon > 0$ be a constant. There is an algorithm with $\tilde O(hn^{3/2}/\sqrt{\ell} + n^{3/2 + \epsilon}/\sqrt{\ell})$
preprocessing time which either reports the existence of a $K_h$-minor of $G$ or which, at any point in the dynamic scenario above,
can decompose each connected component of $G\setminus X$ containing at least one
boundary vertex of $\mathcal C_X$ into $X$-clusters of $\mathcal C_X$ and report the vertex weight of each such component
in a total of  time $\tilde O(h^2\sqrt{\ell n})$.
\end{Lem}
Recall that in Section~\ref{sec:SepThm}, we need to identify $G'$ and its weight in each iteration. We can apply
Lemma~\ref{Lem:DecompSubgraphs} for this since $G'\in\mathcal H$, i.e.,  it contains at least one boundary vertex of
$\mathcal C_X$ (if it did not, we would have found a separator completely contained in a cluster in
Section~\ref{sec:SepThm}).

\section{Finding Small Trees}\label{sec:ShortestPaths}
In this section, we will make use of the decomposition technique of the previous section to find an index $i$ or a tree $T$ satisfying
Lemma~\ref{Lem:TS}. Our algorithm may fail to do this but in such a case, it will find two vertices of $G'$ that are far apart.
In Section~\ref{sec:SepSet}, we show how to obtain from these two vertices a set $S$ satisfying Lemma~\ref{Lem:TS}.

\subsection{Obtaining a sublinear size graph}
To find a tree of small depth intersecting all the $A_i$-sets, we use a technique similar to that of Plotkin, Rao, and Smith
involving shortest paths (see also the proof of Theorem~\ref{Thm:ShallowMinor}).
However, to speed up computations, we shall instead consider approximate shortest paths in a graph of sublinear size that we define
and analyze in the following.

\paragraph{Dense distance graphs:}
For a cluster $C$ and a subset $B$ of $\delta C$, consider the complete undirected graph
$D_B(C)$ with vertex set $B$. Each edge $(u,v)$ in $D_B(C)$ has weight equal to the weight of a shortest path in $C$ between $u$ and
$v$ that does not contain any other vertices of $\delta C$. We call $D_B(C)$ the \emph{dense distance graph of $C$ (w.r.t.\ $B$)}.
By decomposing shortest paths of $C$ at boundary vertices, it follows that $d_{D_B(C)}(u,v) = d_C(u,v)$ for all $u,v\in B$.

We maintain, for each cluster $C\in\mathcal C$, dense distance graph $D_{\delta C\setminus X}(C)$ of the set of
passive boundary vertices of $C$. In a preprocessing step, we compute and store, for each pair of boundary vertices
$u,v\in\delta C$ a shortest path (if any) from $u$ to $v$ in $C$ that does not visit any other boundary vertices (i.e., a shortest
path in $C\setminus(\delta C\setminus\{u,v\})$). For each $u$, we store these paths compactly in
a shortest path tree and we also store the distances from $u$ to each $v$. This takes $O(|C||\delta C|)$
time over all $u$ and $v$ in $\delta C$ which by Lemma~\ref{Lem:R} is $\tilde O(hn^{3/2}/\sqrt{\ell})$ over all clusters
$C\in\mathcal C$.
From the distances computed, we can obtain dense distance graphs $D_{\delta C}(C)$ over all $C$ within the same time bound.
The following lemma shows that dense distance graphs over subsets of $\delta C$ can be obtained more efficiently when we are given
$D_{\delta C}(C)$.
\begin{Lem}\label{Lem:PartialDDG}
For a cluster $C$ and $B\subseteq\delta C$, $D_B(C)$ can be obtained from
$D_{\delta C}(C)$ in $O(|B|^2)$ time.
\end{Lem}
\begin{proof}
By definition of dense distance graphs, $D_B(C)$ is the subgraph of $D_{\delta C}(C)$ induced by $B$.
\end{proof}
At the beginning of the dynamic algorithm and whenever a boundary vertex of $C$ changes its state from
passive to active or vice versa, we invoke Lemma~\ref{Lem:PartialDDG} to update $D_{\delta C\setminus X}(C)$.
Since there are $O(|\delta C|)$ updates to $\delta C$ in total,
Lemma~\ref{Lem:PartialDDG} shows that the total time for maintaining $D_{\delta C\setminus X}(C)$ is
$O(|\delta C|^3)$. Over all clusters $C\in\mathcal C$, this is $\tilde O(h^2n^{3/2}/\sqrt{\ell} + h^4n)$ by Lemma~\ref{Lem:R}.

Consider the graph $D_X$ defined as the union of these dense distance graphs over clusters in $\mathcal C_X$. Since each such
dense distance graph represents shortest paths inside a cluster that avoid exactly the active vertices, a shortest path in
$D_X$ between any two boundary vertices has the same weight as a shortest path between them
in $G\setminus X$. Also note that $D_X$ only has
$\tilde O(h^2\sqrt{\ell n})$ vertices by Lemma~\ref{Lem:RX}. However, finding shortest paths in $D_X$ is no faster than in
$G\setminus X$ since $D_X$ is too dense.
We use spanners to obtain a sparse approximate representation of $D_X$ and we will use this sparse representation to help us find a
tree satisfying Lemma~\ref{Lem:TS}. We need the following result from~\cite{SpannerGeneral}.
\begin{Lem}\label{Lem:Spanner}
Let $H$ be an undirected graph with non-negative edge weights. For any integer
$k\geq 1$, a $(2k - 1)$-spanner of $H$ of size $O(k|V(H)|^{1 + 1/k})$ can be constructed in $O(k|E(H)|)$ time.
\end{Lem}

Fix a constant $\epsilon > 0$ and let $k = \lceil1/\epsilon\rceil$.
For each $C\in\mathcal C$, we keep a $(2k - 1)$-spanner $S(C)$ of
$D_{\delta C\setminus X}(C)$. Whenever $D_{\delta C\setminus X}(C)$ is updated, we invoke Lemma~\ref{Lem:Spanner}
to update $S(C)$ accordingly.

Since there are $O(|\delta C|)$ updates to $\delta C$ during the course of the dynamic algorithm,
Lemma~\ref{Lem:Spanner} implies that the total time for maintaining $S(C)$
is $O(|\delta C|^3)$. Over all clusters, this is $\tilde O(h^2n^{3/2}/\sqrt{\ell} + h^4n)$ by
Lemma~\ref{Lem:R}. Let us denote by $S_X$ the graph obtained as the union of $S(C)$ over all $C\in\mathcal C_X$.

\subsection{Small-depth tree in sublinear time}\label{subsec:ApprSP}
Now let us describe how to find a tree $T$ satisfying Lemma~\ref{Lem:TS}. We will first find a subtree in
$S_X$ and then extend it to a tree in $G$ with the desired properties. We need the following lemma.
\begin{Lem}\label{Lem:SPLayered}
A shortest path tree from any vertex in $S_X$ can be computed in $\tilde O((h^2\sqrt{\ell n})^{1 + \epsilon})$ time.
\end{Lem}
\begin{proof}
Since we have chosen $k = \lceil1/\epsilon\rceil$ in Lemma~\ref{Lem:Spanner}, the number of edges of a spanner $S(C)$ is
$O(|\delta C|^{1 + \epsilon})$.
Summing over all $C\in\mathcal \mathcal C_X$ gives a total of $\tilde O((h^2\sqrt{\ell n})^{1 + \epsilon})$ edges by
Lemma~\ref{Lem:RX}. Applying Dijkstra in $S_X$ then gives the desired time bound.
\end{proof}

\paragraph{Finding a subtree:}
In the following, we shall refer to $\ell$, $p$, $G'$, and sets $A_i$ and $\mathcal A_i$ as defined in Section~\ref{sec:SepThm}.
Let $s$ be a boundary vertex of an $X$-cluster in the connected component $G'$. Since
$s\in S_X$, we can apply Lemma~\ref{Lem:SPLayered} to find a shortest path tree $T$ in $S_X$ rooted at $s$. This tree corresponds
to a tree in $G'$ by replacing edges by their underlying paths that we have associated with the dense distance graphs. The tree
need not intersect all $A_i$-sets but it will intersect some related sets that we define
in the following. We will later extend $T$ to a tree satisfying Lemma~\ref{Lem:TS}.

Introduce, for each cluster $C\in\mathcal C$, subsets
$A_1(C),\ldots,A_p(C)$ of $\delta C$. We define $A_i(C)$ as the set of passive vertices $b\in\delta C$ such that
the $X$-cluster of $C$ containing $b$ intersects $A_i$.
\begin{Lem}\label{Lem:AiAlgo}
After each update of $X\cap\delta C$ in the dynamic scenario, sets $A_1(C),\ldots,A_p(C)$ can be computed in a total of
$O(p|C|)$ time for any $C\in\mathcal C$.
\end{Lem}
See Appendix~\ref{sec:AiAlgo} for a proof.
Since $p\leq h$, it follows from Lemma~\ref{Lem:AiAlgo} that the total time to maintain sets $A_1(C),\ldots,A_p(C)$ is
$O(h|\delta C||C|)$. By Lemma~\ref{Lem:R}, this is $\tilde O(h^2n^{3/2}/\sqrt{\ell})$.

\paragraph{Extending to the desired tree:}
For $i = 1,\ldots,k$, let $A_i'$ be the union of $A_i(C)$ over all clusters $C\in\mathcal C_X$. Note
that $A_i'\subseteq V(S_X)\cap V(G')$ so shortest path tree $T$ in $S_X$ intersects all $A_i'$-sets.

Recall that $S_X$ has stretch $2k-1$. First assume that the distance in $T$ from $s$ to each $A_i'$-set is less than
$8\ell\ln n(2k - 1)$. Form a subtree $T'$ of
$T$ by picking, for each $i$, a shortest path in $T$ from $s$ to a vertex $a_i\in A_i'$ such that $d_T(s,a_i) = O(\ell\log n)$.
Each edge $e$ of $T'$ corresponds to an underlying shortest path in $G$ of length equal to the weight of $e$. Since we
keep these paths with the dense distance graphs, we can form the tree $T_G'$ in $G$ corresponding to $T'$ in
$O(|T_G'|) = O(h\ell\log n)$ time.

Tree $T_G'$ has the depth and size required in Lemma~\ref{Lem:TS} but it need not intersect all $A_i$-sets. However, each
$A_i'$-set is ``close'' to $A_i$ so we can expand $T_G'$ slightly to get this additional property. More precisely, for
$i = 1,\ldots,p$, let $C_i$ be a cluster of $\mathcal C_X$ such that $a_i\in A_i(C_i)$. By definition, the $X$-cluster of $C_i$
containing $a_i$ intersects $A_i$. Since we maintain this $X$-cluster, we can visit all its edges in time proportional to its size
which is $\tilde O(hr)$ by Lemma~\ref{Lem:Sparse}. We extend $T_G'$ into this $X$-cluster to connect it to
$A_i$. Total time for this over all $i$ is $\tilde O(h^2r)$.

Since $\ell = \Omega(\sqrt n/(h\sqrt{\log n}))$, a cluster has $O(r) = O(n/(h^2\ell)) = O(\ell\log n)$ vertices. Thus,
the resulting expanded
tree $T_G'$ will still have the depth and size required by Lemma~\ref{Lem:TS} and it can be found in $\tilde O(h^2\ell\log n)$ time.
The expanded tree will also intersect all $A_i$-sets; however, this
is under the assumption that they are all non-empty. If $A_i = \emptyset$ for some $i$ then also $A_i' = \emptyset$ (and vice versa)
so we can identify this situation within the same time bound.

We assumed above that the distance in $T$ from $s$ to each $A_i'$-set is less than $8\ell\ln n(2k - 1)$. If this is not the case,
we can identify a vertex $t$ of an $X$-cluster in $\mathcal C_X$ which is contained in $G'$ such that
$d_{G'}(s,t)\geq d_{S_X}(s,t)/(2k-1) = d_T(s,t)/(2k-1) \geq 8\ell\ln n$.

We now have an algorithm which either finds an index $i$ such that $A_i = \emptyset$, a tree $T$ in $G'$ of depth $O(\ell\log n)$
and size $O(h\ell\log n)$ with $V(T)\cap A_i\neq\emptyset$ for $i = 1,\ldots,p$, or
a pair of boundary vertices $s$ and $t$ in $\mathcal C_X$ both contained in $G'$ with $d_{G'}(s,t) \geq 8\ell\ln n$.
For any constant $\epsilon > 0$, the algorithm runs in
$\tilde O((h^2\sqrt{\ell n})^{1 + \epsilon})$ time with $\tilde O(h^2n^{3/2}/\sqrt{\ell} + n^{3/2 + \epsilon}/\sqrt\ell  + h^4n)$
preprocessing.


\section{Finding Small Separating Sets}\label{sec:SepSet}
Suppose the above algorithm finds vertices $s$ and $t$ with $d_{G'}(s,t) \geq 8\ell\ln n$.
We will show that we can then find a set $S$ satisfying Lemma~\ref{Lem:TS} within the desired time bound.
Consider the iterative algorithm in~\cite{ShallowMinor} to find $S$.
Initialize vertex set $R = \{s\}$. In the $i$th iteration, augment $R$ with vertices within distance $2$ in $G'$ if
either $R$ grows by a factor of at least $(1 + 1/\ell)$ or $V(G')\setminus R$ shrinks by a factor of at least $(1 + 1/\ell)$.
Otherwise terminate. Let $S = N_{G'}(R)$ be the final set $R$.

Now consider the following variant. Start two searches like the one above, one in $s$
and one in $t$. The two searches run in parallel, i.e., after spending one unit of time for one search, we spend one unit
of time for the other search, and so on. Each search must terminate after less than $2\ell\ln n$ iterations. Let $R_s$ be the final
set in the search from $s$, let $R_t$ be the final set in the search from $t$, and let $S_s = N_{G'}(R_s)$ and $S_t = N_{G'}(R_t)$.
All vertices of $S_s$ have distance at most $4\ell\ln n - 1$ to $s$ and all vertices of $S_t$ have distance at most $4\ell\ln n - 1$
to $t$. Since $d_{G'}(s,t)\geq 8\ell\ln n$, $S_s\cap S_t = \emptyset$.

We stop the parallel algorithm when the first of the two searches terminates, say, the search from $s$. Let $S = S_s$.
We have spent $O(|G'[N_{G'}(S)]|)$ time.
Since $N_{G'}(S_t)\subset N_{G'}(V(G')\setminus S)$, we have
$|G'[N_{G'}(V(G')\setminus S)]| > |G'[N_{G'}(S_t)]|\geq |G'[N_{G'}(S)]|$. Lemma~\ref{Lem:TS} now follows.

\section{Applications}\label{sec:Applications}
Applying Theorem~\ref{theorem:Sep} and Corollary~\ref{Cor:Sep}, we get a separator theorem similar to that
in~\cite{FastSeparatorHMinor} but with a different trade-off and only $O(\poly(h))$ time dependency on $h$
and no dependency in the size.
\begin{theorem}\label{Thm:SepTradeOff}
Let $G$ be a graph with $n$ vertices and let $h\in\mathbb N$. For any constant $\epsilon > 0$ and
for any $3/4\leq\delta < 1$, there is an algorithm which either reports that $G$ has $K_h$ as a minor or
returns a separator of $G$ of size $\tilde O(n^\delta)$.
Its running time is linear plus $\tilde O(h^{15/2}n^{5 - 5\delta + \epsilon} + h^8n^{4 - 4\delta})$.
\end{theorem}
\begin{proof}
By Lemma~\ref{Lem:Sparse}, we may assume that $|G| = \tilde O(hn)$ since otherwise $G$ has $K_h$ as a minor.
Identify in linear time the set $S$ of vertices of degree greater than $hn^{1 - \delta}$. If $S$ is a separator
in $G$ (which we can check in linear time),
we are done since $|S| = \tilde O(n^{\delta})$. Otherwise, it suffices to find an $\tilde O(n^{\delta})$-size separator of
the connected component $G'$ of $G\setminus S$ having $w(G') > \frac 2 3 w(G)$.

Let $\epsilon'\geq 0$ be a constant (to be specified) and let $T$ be a spanning tree of $G'$. In $O(|G'|)$ time, partition $T$ into
$O(h^4n^{1 - \epsilon'})$ subtrees each of size $O(n^{\epsilon' + 1 - \delta}/h^3)$ with the linear time procedure
\texttt{FINDCLUSTERS} of Frederickson~\cite{SpanningForest} but with maximum vertex degree $hn^{1 - \delta}$ instead of degree $3$.

Let $G''$ be the graph obtained from $G'$ by contracting the connected components induced by the subtrees found and removing self-loops
and multiple edges. It has $O(h^4n^{1 - \epsilon'})$ vertices. Assign a weight to each vertex equal to the weight of the
corresponding contracted component in $G'$.

By Corollary~\ref{Cor:Sep}, we can find a separator of $G''$ of size $\tilde O(h^3n^{(1 - \epsilon')/2})$ in
$\tilde O(h^{15/2 + \varepsilon}n^{(1 - \epsilon')(5/4 + \varepsilon)} + h^8n^{1 - \epsilon'})$ time
for any constant $\varepsilon > 0$. Taking
the connected components of $G$ corresponding to separator vertices of $G''$ gives a separator in $G$ of size
$\tilde O(h^3n^{(1 - \epsilon')/2}n^{\epsilon' + 1 - \delta}/h^3) = \tilde O(n^{3/2 - \delta + \epsilon'/2})$. We pick
$\epsilon' = 4\delta - 3$ to obtain a separator in $G$ of the desired size $\tilde O(n^{\delta})$ in time
$\tilde O(h^{15/2 + \varepsilon}n^{5 - 5\delta + (4 - 4\delta)\varepsilon} + h^8n^{4 - 4\delta})$ which is
$\tilde O(h^{15/2}n^{5 - 5\delta + (5 - 4\delta)\varepsilon} + h^8n^{4 - 4\delta})$. Setting
$\epsilon = (5 - 4\delta)\varepsilon$ shows the theorem.
\end{proof}
\begin{Cor}\label{Cor:LinTimeSep}
Let $G$ be a graph with $n$ vertices and let $h\in\mathbb N$. For any constant $\epsilon > 0$,
there is an algorithm which either reports that $G$ has $K_h$ as a minor or
returns a separator of $G$ of size $O(n^{4/5 + \epsilon})$.
Its running time is $O(|G| + h^{15/2}n^{1 - 3\epsilon} + h^8n^{4/5 - 3\epsilon})$.
\end{Cor}

We now get faster algorithms for shortest paths and maximum matching. The following three theorems follow
easily from our results and those in~\cite{SSSPPlanar},~\cite{Yuster}, and~\cite{MaxMatching}, respectively.
\begin{theorem}\label{Thm:SSSP}
Let $G$ be a bounded-degree graph with $n$ vertices and with non-negative edge weights, let $s$ be a vertex in $G$,
and let $h\in\mathbb N$. There is a constant $c > 0$ so that if $h = O(n^c)$, there is an algorithm
with $O(n)$ running time (no dependency on $h$) which either reports that $G$ has $K_h$ as a minor or
returns a shortest path tree in $G$ rooted at $s$.
\end{theorem}
\begin{theorem}\label{Thm:SSSPNeg}
Let $G$ be an edge-weighted graph with $n$ vertices, let $L$ be the absolute value of the smallest edge weight, and let
$h\in\mathbb N$. For any vertex $s$ in $G$, there is an algorithm with $\tilde O(\poly(h)n^{4/3}\log L)$ running time which either
reports that $G$ has $K_h$ as a minor,
reports that $G$ has a negative cycle reachable from $s$, or
returns a shortest path tree in $G$ rooted at $s$.
\end{theorem}
\begin{theorem}\label{Thm:MaxMatching}
Let $G$ be a graph with $n$ vertices and non-negative edge weights. Let $h\in\mathbb N$ and let $\epsilon > 0$ be an arbitrarily
small constant. There is an $O(\poly(h)n^{\omega(3/2 + \epsilon)/(\omega + 1/2)}) \approx O(\poly(h)n^{1.239})$-time algorithm
reporting that $G$ has $K_h$ as a minor or giving a maximum matching of $G$; here $\omega < 2.376$ is the matrix multiplication
constant.
\end{theorem}

%

\appendix

\section{Proof of Theorem~\ref{Thm:ShallowMinor}}\label{sec:ProofShallowMinorAll}
We need the following lemma which can be regarded as a generalization of Lemma 2.5 in~\cite{ShallowMinor}.
\begin{Lem}\label{Lem:TreeOrCut}
Let $G = (V,E)$ be an unweighted graph with $n$ vertices, let $A_1,\ldots,A_p$ be non-empty subsets of $V$, and let
$\ell,\delta\in\mathbb N$. Then either
\begin{enumerate}
\item there is a tree $T$ in $G$ of depth at most $4\delta \ell\ln n$ and with at most $4\delta\ell p\ln n$ vertices such that
      $V(T)\cap A_i\neq\emptyset$ for $i = 1,\ldots,p$, or
\item there exists a set $S\subset V$ where
  \begin{enumerate}
  \item $|N^\delta(S)\cap V\setminus S| < \min\{|S|,|V\setminus S|\}/\ell$, and
  \item $|N^\delta(V\setminus S)\cap S| < \min\{|S|,|V\setminus S|\}/\ell$.
  \end{enumerate}
\end{enumerate}
\end{Lem}
\begin{proof}
The proof is similar to that in~\cite{ShallowMinor}. Initialize subset $R$ of $V$ to $\{v\}$ for some node $v\in V$.
In each step, add to $R$ vertices that are within distance $2\delta$ of $R$. The process continues as long as $|R|$ grows
by a factor of at least $1 + 1/\ell$ or $|V\setminus R|$ shrinks by a factor of at least $1 + 1/\ell$.

This process must stop after at most $2\ell\ln n$ iterations since $(1 + 1/\ell)^{2\ell\ln n} > n$. If $R = V$ at this point, the
length of a path in $G$ from $v$ to each of the sets $A_i$ is bounded by $4\delta\ell\ln n$. Hence, the first condition of the lemma
is satisfied.

Now assume that $R\neq V$ at termination. Let $S = N^\delta(R)$. We have
\[
  |N^{\delta}(S)\cap V\setminus S| = |N^{2\delta}(R)\cap V\setminus S| < |N^{2\delta}(R)\cap V\setminus R|
\]
and
\[
  |N^{\delta}(V\setminus S)\cap S| = |N^{\delta}(R)\cap V\setminus R| < |N^{2\delta}(R)\cap V\setminus R|.
\]
Since $R\neq V$, we must have
\begin{align*}
|N^{2\delta}(R)\cap V\setminus R| < \min\{|R|,|V\setminus N^{2\delta}(R)|\}/\ell < \min\{|S|,|V\setminus S|\}/\ell,
\end{align*}
and it follows that the second condition of the lemma is satisfied.
\end{proof}

Now, let us prove Theorem~\ref{Thm:ShallowMinor}.
The algorithm is similar to that in~\cite{ShallowMinor} but we use a dynamic spanner to speed up computations.
\paragraph{Invariants:} A four-way
partition $(V_r, M, B, V')$ of $V$ is maintained. Initially, $V' = V$ and $M = B = V_r = \emptyset$. The following invariants are
satisfied at the beginning of each iteration:
\begin{enumerate}
\item Set $M$ can be partitioned into $p\leq h$ subsets $\mathcal A_1,\ldots,\mathcal A_p$ each of size $O(\ell h\log n)$. The diameter
      of each of the induced subgraphs is $O(\ell\log n)$. For each pair of these subgraphs, there is an edge in $G$ between them,
      certifying that $K_p$ is a minor of $G$ of depth $O(\ell\log n)$,
\item $|B|\leq |V_r|/\ell$,
\item $w(V_r)\leq 2w(V)/3$.
\end{enumerate}
\paragraph{The algorithm and correctness:}
In each iteration, the algorithm either augments $M$ with an additional set to form a larger minor, or adds nodes from $M$ and/or
nodes from $V'$ to $V_r\cup B$. If $p = h$, the algorithm terminates and outputs a $K_h$-minor of $G$ of depth $O(\ell\log n)$ formed
by sets $\mathcal A_1,\ldots,\mathcal A_h$.

Let $w$ be the vertex weight function. If $w(V(C))\leq \frac 2 3 w(V)$ for each connected component $C$ of $G[V']$
the algorithm also terminates and returns the separator $M\cup B$.

The invariants imply that the algorithm gives the correct output at termination. We need to show that the invariants
are maintained.

Initially, $V' = V$ and the invariants are trivially satisfied. Now assume the invarants are satisfied at the beginning of an
iteration and assume the algorithm does not terminate in this iteration. We will show that the invariants are still satisfied at
the beginning of the next iteration. Note that since the algorithm does not terminate, we have $w(V') > \frac 2 3 w(V)$,
implying that $w(V_r) < w(V)/3$.

Let $G'$ be the unique connected component in $G[V']$ of weight greater than $\frac 2 3 w(V)$.
For $i = 1,\ldots,p$, let $A_i = N_G(\mathcal A_i)\cap V(G')$. If some $A_i$ is empty, we move
$\mathcal A_i$ from $M$ to $V_r$ and start the next iteration. In this case, the first and second invariants are trivially
still satisfied and the third invariant is satisfied since $w(V\setminus V') < w(V)/3$.

Now suppose all $A_i$-sets are non-empty. Let $k = \lceil 2/\epsilon\rceil$ and let $\delta = 2k - 1$.
We keep a $\delta$-spanner $H$ of
$G'$ and we apply Lemma~\ref{Lem:TreeOrCut} to $H$. If a tree $T$ is found, we move its vertex set
$\mathcal A_{p + 1}$ from $V'$ to $M$. Since $\delta = O(1)$, the invariants are clearly satisfied in the beginning of the next
iteration. We may assume that $T$ has size $\Theta(\min\{h\ell\log n,|V(G')|\})$ since otherwise, we can expand it in $G'$ to a
tree of this size.

Now assume that a set $S$ is found in the application of Lemma~\ref{Lem:TreeOrCut}. Let $S'$ be the set of smaller weight among
$S$ and $V(G')\setminus S$ and let $B' = N_{G'}(S')\setminus S'$. The algorithm moves $S'$ to $V_r$, moves $B'$ to $B$, and
starts the next iteration. To show that the invariants are still
satisfied, we will show below that $B'\subseteq B_{\delta}$, where
$B_{\delta} = N_H^\delta(S')\cap V(G')\setminus S'$. Assuming
this, it follows from Lemma~\ref{Lem:TreeOrCut} that at most
$|B_{\delta}| < |S'|/\ell$ vertices are added
to $B$ when $S'$ is added to $V_r$. Since $|S'|$ vertices are added to $V_r$, the second invariant is maintained. Let
$V_r^{(1)}$ resp.\ $V_r^{(2)}$ be the set $V_r$ before resp.\ after $S'$ is moved to it. As observed above, $w(V_r^{(1)}) < w(V)/3$ so
$w(V_r^{(2)})$ is bounded by
\[
  w(V_r^{(1)}) + w(S')\leq w(V_r^{(1)}) + w(V')/2 \leq (w(V_r^{(1)}) + w(V))/2 < 2w(V)/3,
\]
where $V'$ is the set before $S'$ is moved. Hence the third invariant is maintained.

What remains is to show that $B'\subseteq B_{\delta}$. Let $v\in B'$.
There is an edge $(u,v)\in E$ with $u\in S'$.
Since $H$ is a $\delta$-spanner of $G'$, $d_H(u,v)\leq\delta d_{G'}(u,v) = \delta$ and hence $v\in N_H^\delta(S')$.
Since also $v\in V(G')\setminus S'$, we have $v\in B_{\delta}$. This shows the desired.

\paragraph{Running time:}
We maintain $H$ using the dynamic spanner for unweighted
graphs of Baswana and Sarkar~\cite{DynGraphSpanner}. This spanner can be constructed in linear time and requires
$O(7^{\delta/4}) = O(1)$ update time per edge removal. In an iteration, if the vertex set of a tree $T$ is moved to $M$,
we delete from $G'$
all edges of $T$ and edges incident to this tree. If instead set $S$ is found, we delete from $G'$ all edges incident to
$S'\cup B'$. If $G'$ becomes disconnected, we redefine it to be a connected component of maximum weight and delete all other
components. Since we never add anything to $G'$ during the course of the algorithm, we use a total of $O(m)$
time for maintaining $H$.

The spanner in~\cite{DynGraphSpanner} consists of $O(|V(G')|^{1 + 1/k}\log^2|V(G')|)$ edges. Hence, each iteration can be executed in
$O(n^{1 + 1/k}\log^2n)$ time in addition to the time for edge removals. In each iteration, $\Omega(\ell)$ vertices are moved from
$V'$ or from $M$ so the number of iterations is
$O(n/\ell)$. Hence, total time is $O(m + n^{2 + \epsilon/2}(\log^2n)/\ell) = O(m + n^{2 + \epsilon}/\ell)$, as desired.

\section{Proof of Lemma~\ref{Lem:rDiv}}\label{sec:ProofrDiv}
Below we will bound the time to find an $r$-clustering. Given such a clustering,
we can apply Corollary~\ref{Cor:ShallowMinor} with $\epsilon/2$ instead of $\epsilon$
recursively to each cluster to
form a nested $r$-clustering (or a $K_h$-minor). This is done in such a way that sizes of clusters and their
number of boundary vertices go down geometrically along any root-to-leaf path in the recursion. More precisely, for a cluster $C$,
associate two vertex weight functions, $w_1$ and $w_2$. Let $w_1 = 1/|V(C)|$ and let $w_2$ be $1/|\delta(C)|$ on $\delta C$ and
$0$ on $V(C)\setminus\delta(C)$. Now apply Theorem $5$ in~\cite{MultiWeightSep} with Corollary~\ref{Cor:ShallowMinor} as the
separator theorem. Assuming no $K_h$-minor of $G$ is found, this partitions
$C$ into children each containing at most $c_1|V(C)|$ vertices and at most $c_2|\delta(C)| + \tilde O(h\sqrt{|V(C)|})$ boundary
vertices, for constants $c_1,c_2 < 1$.

Let $c = \max\{c_1,c_2\}$. Suppose $C$ is a level $i$-cluster and let $C'$ be its ancestor level $1$-cluster. It follows
from the above that $C$ has $O(c_1^i|V(C')|)$ vertices and
\begin{align*}
  \tilde O(h\sum_{1\leq j\leq i}c_2^{i-j}\sqrt{c_1^j|V(C')|}) & = \tilde O(h\sqrt{|V(C')|}\sum_{1\leq j\leq i}c^{i - j/2})\\
                                                              & = \tilde O(h\sqrt{|V(C')|}c^{i/2}\sum_{0\leq j < i}c^{j/2})\\
                                                              & = \tilde O(hc^{i/2}\sqrt{|V(C')|})
\end{align*}
boundary vertices. Hence, both sizes of clusters and their number of boundary vertices go down geometrically, as desired.

The worst case time for this recursive partition occurs when clusters in the $r$-clustering (i.e., level $1$-clusters) each contain
$r$ vertices. In this case, there are $O(n/r)$ level $1$-clusters and the total time to partition them all into their children is
$O(m + h(n/r)r^{3/2 + \epsilon/2}) = O(m + hnr^{1/2 + \epsilon/2})$. This time dominates the time for partitioning at each recursion
level so the total time is within that stated in the lemma.

\paragraph{Weak $r$-clustering algorithm and correctness:}
Define a \emph{weak $r$-clustering} of $G$ as a clustering with clusters having a total of $\tilde O(hn/\sqrt r)$ boundary
vertices (counted with multiplicity) and each containing at most $r$ vertices (so compared to an $r$-clustering, we drop the bound
on the number of boundary vertices per cluster). We will first find a weak $r$-clustering and then obtain an $r$-clustering from it.

Let $\varepsilon = \epsilon/3$. We apply Theorem~\ref{Thm:ShallowMinor} with
$\ell = n^{1 - \varepsilon}/(hr^{1/2 - \varepsilon}\sqrt{\log n})$,
with $\varepsilon$ instead of $\epsilon$, and with vertex weights evenly distributed. Since
$n/\ell = n^{\varepsilon}hr^{1/2 - \varepsilon}\sqrt{\log n} \leq h^2\ell\log n$,
this gives a separator of size $O(hn^{1 - \varepsilon}\sqrt{\log n}/r^{1/2 - \varepsilon})$ in
$O(m + hn^{1 + 2\varepsilon}r^{1/2 - \varepsilon}\sqrt{\log n})$ time that partitions the graph into connected components
(or a $K_h$-minor of $G$ is identified). Now recurse on each of
them. In the general step, for a connected subgraph with $\hat n$ vertices, apply
Theorem~\ref{Thm:ShallowMinor} with $\ell = \hat n^{1 - \varepsilon}/(hr^{1/2 - \varepsilon}\sqrt{\log n})$. The recursion
stops when a subgraph with at most $r$ vertices is obtained.

We claim that this algorithm gives a weak $r$-clustering. First, let us bound the total number of boundary vertices, $B(n)$,
introduced by recursively partitioning $G$ (here we count them with multiplicity). The proof
is quite similar to that of Frederickson (see the proof of Lemma 1 in~\cite{APSPPlanar}). We claim that
\[
  B(n) \leq \frac {h\sqrt{\log n}}{\sqrt r}(cn - dr^{\varepsilon}n^{1 - \varepsilon}),
\]
for some constants $c,d > 0$. The proof is by induction on $n$. The inequality is clearly satisfied when $n = \Theta(r)$ for $c$
sufficiently larger than $d$. Now consider the induction step and assume the claim holds for smaller values. In
the analysis, we may assume the worst case situation where a cluster is split into exactly two sub-clusters. Then for values
$\alpha_1,\alpha_2\in[1/3,2/3]$, $\alpha_1 + \alpha_2 = 1$, we have
\begin{align*}
  B(n) & \leq \frac{c'hn^{1 - \varepsilon}\sqrt{\log n}}{r^{1/2 - \varepsilon}} +
              B\left(\alpha_1n + \frac{c'hn^{1 - \varepsilon}\sqrt{\log n}}{r^{1/2 - \varepsilon}}\right) +
              B\left(\alpha_2n + \frac{c'hn^{1 - \varepsilon}\sqrt{\log n}}{r^{1/2 - \varepsilon}}\right)\\
       & <    \frac{c'hn^{1 - \varepsilon}\sqrt{\log n}}{r^{1/2 - \varepsilon}} +
              \frac {h\sqrt{\log n}}{\sqrt r}\left(cn + \frac{2cc'hn^{1 - \varepsilon}\sqrt{\log n}}{r^{1/2 - \varepsilon}} -
              dr^{\varepsilon}(\alpha_1n)^{1 - \varepsilon} - dr^{\varepsilon}(\alpha_2n)^{1 - \varepsilon}\right),
\end{align*}
for some constant $c'$. It suffices to show that
\[
  c' + \frac{2cc'h\sqrt{\log n}}{\sqrt r}\leq d\left(\alpha_1^{1 - \varepsilon} + \alpha_2^{1 - \varepsilon} - 1\right)
\]
Since $\alpha_1,\alpha_2\in[1/3,2/3]$ and $\alpha_1 + \alpha_2 = 1$, the right-hand side is $d\delta$ for some constant
$\delta > 0$. Since $r > Ch^2\log n$, we can satisfy the inequality for sufficiently large $C$ and $d$ (while keeping $c$ sufficiently
larger than $d$ as required above). This shows the bound on the number of boundary vertices.

Correctness of the algorithm now follows: there are $O(hn\sqrt{\log n}/\sqrt r)$ boundary vertices in total and each cluster has at
most $r$ vertices by construction. Hence, the algorithm gives a weak $r$-clustering for $G$ (or a $K_h$-minor).

\paragraph{Running time of weak $r$-clustering algorithm:}
Let $T(m,n)$ be the running time. We claim that
\[
  T(m,n) \leq k(m + hn^{1 + 2\varepsilon}r^{1/2 - \varepsilon}\sqrt{\log n})\log n,
\]
for some constant $k$. This will show the desired since $2\varepsilon = 2\epsilon/3$.
The proof is by induction on $n$. The bound holds for $n = \Theta(r)$. Now consider the induction
step and assume the bound holds for smaller values. For some constant $k' > 0$, we have
\[
  T(m,n) \leq k'(m + hn^{1 + 2\varepsilon}r^{1/2 - \varepsilon}\sqrt{\log n}) + T(m_1,n_1) + T(m_2,n_2),
\]
where $m_1 + m_2 = m$ and $n_1,n_2 < c'n$, $c' < 1$ constant. The induction hypothesis gives
\begin{align*}
  T(m,n) & \leq k'(m + hn^{1 + 2\varepsilon}r^{1/2 - \varepsilon}\sqrt{\log n}) +
                k(m + hr^{1/2 - \varepsilon}(n_1^{1 + 2\varepsilon} + n_2^{1 + 2\varepsilon})\sqrt{\log n})\log(c'n)\\
         & <    (m + hn^{1 + 2\varepsilon}r^{1/2 - \varepsilon}\sqrt{\log n})(k' + k\log(c'n))\\
         & =    (m + hn^{1 + 2\varepsilon}r^{1/2 - \varepsilon}\sqrt{\log n})(k\log n + k' + k\log c').
\end{align*}
Since $\log c' < 0$, we can pick $k$ sufficintly larger than $k'$ to get $k' + k\log c'\leq 0$. This shows the time bound.

\paragraph{$r$-clustering algorithm:}
For each cluster $C$ in the weak $r$-clustering which does not contain $\tilde O(h\sqrt{r})$ boundary vertices, we apply
Corollary~\ref{Cor:ShallowMinor} with $\epsilon/2$ instead of $\epsilon$ and with vertex-weights evenly distributed on the
\emph{boundary vertices} of $C$. We repeat this
process until each cluster contains $\tilde O(h\sqrt r)$ boundary vertices. We will show that this gives the desired $r$-clustering.

By construction, each cluster has $O(r)$ vertices and $\tilde O(h\sqrt r)$ boundary vertices. To show that the total number of
boundary vertices is $\tilde O(hn/\sqrt r)$, consider a cluster $C$ in the weak $r$-clustering that contains $\tilde\Omega(h\sqrt r)$
boundary vertices. We will show that the repeated splitting of it introduces only $\tilde O(h\sqrt r)$ additional boundary
vertices. It will then follow that the total number of boundary vertices is $\tilde O(hn/\sqrt r)$.

The worst case occurs when Corollary~\ref{Cor:ShallowMinor} gives an unbalanced split of the vertices of $C$ (this may happen since
we partition w.r.t.\ boundary vertices). Let $C'$ be the largest
sub-cluster obtained. We can ignore the contribution from the smaller sub-clusters in the analysis. Sub-cluster $C'$ has a constant
$c < 1$ fraction of the boundary vertices of $C$. We can repeat this argument recursively on $C'$. At each step, we introduce only
$\tilde O(h\sqrt{|V(C)|}) = \tilde O(h\sqrt r)$ new boundary vertices. After $O(\log|\delta C|)$ steps, the
process stops and we have introduced only $\tilde O(h\sqrt{r})$ boundary vertices in total, as desired.

As for running time, it is easy to see that the worst case occurs when each cluster $C$ of the weak $r$-clustering that is split has
the maximum of $r$ vertices and when each split is unbalanced. There are $O(\log|\delta C|) = O(\log r)$ splits of big sub-clusters
of $C$ and each takes $O(|E(C)| + h|V(C)|^{3/2 + \epsilon/2})$ time to perform. Since the time to split big sub-clusters dominates,
total time spent on splitting $C$ is $O(|E(C)|\log r + h|V(C)|^{3/2 + \epsilon/2}\log r)$.
There are $O(n/r)$ clusters of size $r$ so total time is $O(m\log r + hnr^{1/2 + \epsilon/2}\log r)$ which is within the bound in the
lemma.

\section{Proof of Lemma~\ref{Lem:R}}\label{sec:ProofR}
In the analysis, we may assume that in $\mathcal C$, each non-leaf cluster has exactly two children since this is the worst-case
in terms of maximizing the two sums.

We will first prove that $\sum_{C\in\mathcal C}|V(C)||\delta C| = \tilde O(n^{3/2}/\sqrt{\ell})$. By Lemma~\ref{Lem:Sparse}, this will
imply the first sum.
By construction of $\mathcal C$, there is a constant $c < 1$ such that for any cluster $C\in\mathcal C$, its children each have at
most $c|V(C)|$ vertices and at most $c|\delta C|$ boundary vertices.
Distribute clusters of $\mathcal C$ into a logarithmic number of groups where the $i$th group $\mathcal C_i$ consists of the clusters
containing at least $(1/c)^i$ and less than $(1/c)^{i+1}$ vertices, $i\geq 0$. Each cluster belongs to exactly one group and by
definition of $c$, if a cluster belongs to a group $\mathcal C_i$ then none of its ancestors or descendants from $\mathcal C$
belong to $\mathcal C_i$. Hence, clusters of $\mathcal C_i$ share only boundary vertices. The same analysis as that of
Frederickson (see proof of Lemma $1$ in~\cite{APSPPlanar}) shows that the total number of boundary vertices in clusters of
$\mathcal C_i$ is $\tilde O(hn/\sqrt{(1/c)^i}) = \tilde O(hnc^{i/2})$. Thus the clusters of
$\mathcal C_i$ contribute with a total value of $\tilde O(h(1/c)^inc^{i/2}) = \tilde O(hn(1/c)^{i/2})$ to the sum
$\sum_{C\in\mathcal C}|V(C)||\delta C|$. This is $\tilde O(n^{3/2}/\sqrt{\ell})$ since $(1/c)^i = O(r) = O(n/(h^2\ell))$. This gives
the desired bound on $\sum_{C\in\mathcal C}|V(C)||\delta C|$ since there is only a logarithmic number of groups, and the first
part of the lemma follows.

Now let us show the second sum. Consider a non-leaf cluster $C\in\mathcal C$ and let $P$ be the separator
that partitions $C$ into its children $C_0$ and $C_1$. Let $\beta$ be the constant of the separator theorem used to partition
clusters of $\mathcal C$. Then $|P|\leq\beta h\sqrt{|V(C)|\log|V(C)|}$.
Let $i\in\{0,1\}$. By Theorem $5$ in~\cite{MultiWeightSep} we may assume that
$|V(C_i)\setminus P|\leq |V(C)|/2$ and $|\delta C_i\setminus P|\leq|\delta C|/2$.
Then $|V(C_i)|\leq |V(C)|/2 + \beta h\sqrt{|V(C)|\log|V(C)|}$ and $|\delta C_i|\leq |\delta C|/2 + \beta h\sqrt{|V(C)|\log|V(C)|}$.
Since $|V(C_{1-i})\setminus P|\leq |V(C)|/2$ and $|\delta C_{1-i}\setminus P|\leq|\delta C|/2$,
we also have lower bounds $|V(C_i)| = |V(C)| - |V(C_{1-i})\setminus P|\geq |V(C)|/2$
and $|\delta C_i| = |\delta C| - |\delta C_{1-i}\setminus P|\geq |\delta C|/2$.
Thus, by adding at most order $h\sqrt{|V(C)|\log|V(C)|}$ dummy vertices to $V(C_i)\setminus \delta C_i$ and/or $\delta C_i$ if
necessary, we may assume that $|V(C_i)| = |V(C)|/2 + \beta h\sqrt{|V(C)|\log|V(C)|}$ and
$|\delta C_i| = |\delta C|/2 + \beta h\sqrt{|V(C)|\log|V(C)|}$ (the dummy vertices are only added in the
analysis). By a similar argument,  we may assume that level $1$-clusters all contain the same number of vertices/boundary vertices.

From these simplifying assumptions, it follows that for any two clusters $C,C'\in\mathcal C$, if $|V(C)| = \Theta(|V(C')|)$ then
$|\delta C| = \Theta(|\delta C'|)$. Since clusters of $\mathcal C_i$ have the same number
of vertices (up to constant factor $c$), they thus have the same number of boundary vertices. As we saw above, there are
$\tilde O(hnc^{i/2})$ boundary vertices in total in clusters of $\mathcal C_i$. Counting interior vertices and boundary vertices
separately, we get that the total size of these clusters is $\tilde O(n + hnc^{i/2})$ so the number of clusters is
$\tilde O(nc^i + hnc^{3i/2})$.

First assume that $h\leq (1/c)^{i/2}$. Then there are $\tilde O(nc^i)$ clusters in $\mathcal C_i$ each having $\tilde O(h(1/c)^{i/2})$
boundary vertices. These clusters therefore contribute with a value of $\tilde O(nc^i(h(1/c)^{i/2})^3) = \tilde O(h^3n(1/c)^{i/2})$
to the second sum of the lemma. Since $(1/c)^i = O(n/(h^2\ell))$, this value is $\tilde O(h^2n^{3/2}/\sqrt{\ell})$, as desired.

Now assume that $h > (1/c)^{i/2}$. Then there are $\tilde O(hnc^{3i/2})$ clusters in $\mathcal C_i$ each having
$\tilde O((1/c)^i)$ boundary vertices. These clusters contribute with a value of
$\tilde O(hnc^{3i/2}(1/c)^{3i}) = \tilde O(hn(1/c)^{3i/2}) = \tilde O(h^4n)$ to the second sum. This completes the proof.

\section{Proof of Lemma~\ref{Lem:RX}}\label{sec:ProofRX}
We only show the bound on the sum as the rest is trivial.
For each $x\in X$, let $C_x$ be the cluster with the deepest level in $\mathcal C$ having $x$ as an interior vertex. Let
$\mathcal P_x$ be the set of clusters on the path in $\mathcal C$ from $x$ to the level $1$-cluster containing $C_x$.
Let $\mathcal C_X'$ be the set of clusters
along all these paths as well as their children and let $\mathcal C_1$ be the set of level $1$-clusters of $\mathcal C$.
Then it follows from the procedure to construct $\mathcal C_X$ that $\mathcal C_X = \mathcal C_X'\cup\mathcal C_1$.

Let us bound the number of boundary vertices of clusters in $\mathcal C_X$. The total number of boundary vertices in
$\mathcal C_1$ is $\tilde O(hn/\sqrt r) = \tilde O(h^2\sqrt{\ell n})$. Consider the path in $\mathcal C$ corresponding to a set
$\mathcal P_x$. Since the number of boundary vertices in clusters goes down geometrically along this path, it follows that the total
number of boundary vertices in clusters of $\mathcal P_x$ and their children is bounded by $\tilde O(h\sqrt{|V(C)|})$, where $C$ is the
level $1$-cluster containing clusters of $\mathcal P_x$. Since a level $1$-cluster contains $O(r) = n/(h^2\ell)$ vertices, we conclude
that the clusters of $\mathcal C_X'$ contain a total of $\tilde O(|X|\sqrt{n/\ell})$ boundary vertices. Since
$\mathcal C_X = \mathcal C_X'\cup\mathcal C_X''$ and $X = \tilde O(h^2\ell)$, the lemma follows.

\section{Proof of Lemma~\ref{Lem:AiAlgo}}\label{sec:AiAlgo}
For $i = 1,\ldots,p$, define $\delta_i C$ to be the set of active boundary vertices of $C$ belonging to set $\mathcal A_i$. We can
obtain these sets in $O(|\delta C|)$ time since we represent the $\mathcal A_i$-sets explicitly. For each $i$, we do
as follows. First, mark all vertices of $C$ as not
visited. Now, for each $v\in \delta_i C$, make a DFS from $v$ in $C$ but backtrack if a vertex has already been visited by another
search or if an active vertex other than $v$ is reached. When a passive boundary vertex is visited, add it to $A_i(C)$.

Clearly, this algorithm computes $A_i(C)$ in time $O(|C|)$. To see its correctness, note that a DFS from $v$ will only
visit $X$-clusters of $C$ that are incident to $v$ since we backtrack if another active vertex is reached.
Hence, each passive boundary vertex $b$ added to $A_i(C)$ is contained in an $X$-cluster incident to an
active vertex belonging to $\mathcal A_i$, namely $v$. Said differently, $b$ is contained in an $X$-cluster of $C$ intersecting $A_i$.
If a boundary vertex $b$ contained in such an $X$-cluster is not visited by the DFS from
$v$, it must have been visited by an earlier DFS in the $i$th iteration. Hence, it was added to $A_i(C)$ during that search. This
shows that $A_i(C)$ is correctly computed.

\section{Proof of Theorem~\ref{Thm:SSSP}}\label{sec:ProofSSSP}
Note that $G$ is sparse since it has bounded degree.
By Corollary~\ref{Cor:LinTimeSep}, we can find a separator of $G$ of size $O(n^{4/5 + \varepsilon})$
in $O(n + h^{15/2}n^{1 - 3\varepsilon} + h^8n^{4/5 - 3\varepsilon})$ time for any
constant $\varepsilon > 0$ (or report that $G$ has $K_h$ as a minor). It follows from~\cite{SSSPPlanar}
that if we choose $\varepsilon < 1/5$ then a shortest path tree in $G$ rooted at $s$ can be found in
$O(n)$ additional time since $G$ is sparse. This completes the proof.

The bounded degree assumption is needed here, as noted in~\cite{LinTimeSSSPHMinor}.

\section{Proof of Theorem~\ref{Thm:SSSPNeg}}\label{ProofSSSPNeg}
For a graph $G$ with $n$ vertices and for a parameter $r\in(0,n)$, an \emph{$r$-division} of $G$ is a partition of the edges of
$G$ into $O(n/r)$ groups such that the subgraphs induced by these groups each have $O(r)$ vertices of which only $O(\sqrt r)$
may be incident to edges in other groups.

Yuster~\cite{Yuster} gave an $\tilde{O}(\max\{T(n,h), \poly(h)n^{4/3}\log L\})$-time algorithm for the problem,
where $T(n,h)$ is the time to compute an $n^{2/3}$-division of a graph with $n$ vertices or report that it has $K_h$ as a minor.
Applying Corollary~\ref{Cor:Sep} recursively as in Frederickson's algorithm~\cite{APSPPlanar},
we get $T(n,h) = \tilde O(\poly(h)n^{5/4 + \epsilon})$ for an arbitrarily small constant $\epsilon > 0$.
Corollary~\ref{Cor:Sep} may report that $G$ has $K_h$ as a minor during this construction. This completes the proof.

Yuster did not state the dependency on $h$ in the running time but it is easy to see that it is only polynomial.
In the proof, we could also apply the separator theorem in~\cite{SepOpt}. This would give the same
dependency on $n$ and $L$ but a tower power-dependency on $h$.

\section{Proof of Theorem~\ref{Thm:MaxMatching}}\label{ProofMaxMatching}
A maximum matching of $G$ can be found in time $O(\max\{T(n,h,\delta),\poly(h)n^{\omega\delta}\})$,
where $T(n,h,\delta)$ is the time to find a separator of size $O(n^\delta)$ in an $n$-vertex graph or report that the graph has a
$K_h$-minor~\cite{MaxMatching}. A suitable choice of $\ell$ in Theorem~\ref{theorem:Sep} now shows the desired.

\end{document}